\documentclass[a4paper]{amsart}

\usepackage{amsmath}
\usepackage{amssymb}
\usepackage{amsthm}
\usepackage{bbm}

\newtheorem{theorem}{Theorem}

\newtheorem{lemma}[theorem]{Lemma}

\newcommand{\Oh}{\mathrm{O}}

\title{Can there be an explicit formula for implied volatility?}

\author{Stefan Gerhold}

\address{Vienna University of Technology, Institute of Mathematical Methods in Economics,
Wiedner Hauptstr.\ 8/105-1,
A-1040 Vienna, Austria}

\email{sgerhold at fam.tuwien.ac.at}

\date{\today}

\begin{document}

\begin{abstract}
  It is ``well known'' that there is no explicit expression for the 
  Black-Scholes implied volatility. We \emph{prove} that,
  as a function of underlying, strike, and call price, implied volatility
  does not belong
  to the class of $D$-finite functions. This does not rule out all
  explicit expressions, but shows that implied volatility does not
  belong to a certain large class, which contains many elementary functions
  and classical special functions.
\end{abstract}



\maketitle

\section{Introduction}

The Black-Scholes formula is the most popular way to price European
stock options. Assuming a frictionless market and an underlying
modelled by geometric Brownian motion with volatility~$\sigma$,
the arbitrage free price
of a European call option with strike~$K$, maturity~$T$, and
initial stock price~$S$ is
\[
  C_{\mathrm{BS}}(S,K,T,\sigma) = S N(d_1) - K N(d_2),
\]
where $N$ is the cdf of the standard normal distribution, and
\[
  d_{1,2} = \frac{\log (S/K)}{\sigma \sqrt{T}} \pm
    \frac{\sigma \sqrt{T}}{2}.
\]
(We assume a zero interest rate for simplicity.) Given a call price
surface $C$ that satisfies the no-arbitrage bounds
\begin{equation}\label{eq:arb bds}
  (S-K)^+ \leq C(S,K,T) < S,
\end{equation}
the Black-Scholes formula can be inverted w.r.t.~$\sigma$
to obtain the \emph{implied volatility} $\sigma_{\mathrm{imp}}$:
\begin{equation}\label{eq:imp eq}
  C_{\mathrm{BS}}(S,K,T,\sigma_{\mathrm{imp}}(S,K,T)) = C(S,K,T).
\end{equation}
Implied volatility allows to compare option prices for different strikes, maturities,
underlyings, and valuation times. For entry points to the literature
on implied volatility, see, e.g., \cite{FoJa09,FrGeGuSt11,Ga06,Le05,ScWi08}.

The compute the implied volatility from a call price,
equation~\eqref{eq:imp eq} has to be solved numerically; many
authors have added that the inversion cannot be done in closed form.
The latter statement has been made so often that one might ask whether it
can be turned into something tangible. First, note that it is hard to say
anything about explicit expressions for the function
\begin{equation}\label{eq:not useful}
  (S,K,T) \mapsto \sigma_{\mathrm{imp}}(S,K,T).
\end{equation}
This function actually depends on~$C(\cdot,\cdot,\cdot)$,
and it would be unnatural to make assumptions on
the generic call price surface~$C$ that would
allow to refute a closed form for~\eqref{eq:not useful}. Instead, we consider
the function~$I$ which satisfies
\begin{equation}\label{eq:def I}
  C_{\mathrm{BS}}(S,K,I(S,K,c)) = c, \quad S,K >0,\ (S-K)^+< c < S,
\end{equation}
and is defined on the open set
\[
  \mathcal{D}_I := \{ (S,K,c) \in \mathbb{R}^3 : S,K >0,\ (S-K)^+
    < c < S \}.
\]
The maturity~$T$ has been omitted here, and we will continue to do so,
assuming a fixed $T>0$ throughout.
Note that any kind of explicit expression for~$I$ with variable~$T$ would
imply an explicit expression for any fixed~$T$.
The function~$I$ clearly exists and is unique, since the Black-Scholes
call price increases from the lower to the upper no-arbitrage bound
as~$\sigma$ increases.
If a call price surface~$C$ satisfies the slightly strengthened bounds
\begin{equation}\label{eq:strict}
  (S-K)^+ < C(S,K) < S,
\end{equation}
which usually hold for $T>0$, then the corresponding implied volatility
$\sigma_{\mathrm{imp}}$ satisfies
\[
  \sigma_{\mathrm{imp}}(S,K) = I(S,K,C(S,K)), \quad S,K>0.
\]
An example where the lower bound in~\eqref{eq:strict} fails would be
an out-of-the-money call that cannot move into-the-money, as may happen
in a binomial model, say. We impose~\eqref{eq:strict} instead
of~\eqref{eq:arb bds} to make the domain of~$I$ open, so that~$I$
is differentiable on its whole domain~$\mathcal{D}_I$.

Since the Black-Scholes call price is real analytic for $S,K,\sigma>0$,
and the Black-Scholes vega $\partial C_{\mathrm{BS}}/ \partial \sigma$
is positive, the implicit function theorem~\cite{KrPa02} shows that~$I$
is real analytic (hence $C^\infty$) on~$\mathcal{D}_I$. The question we now ask is whether
the function~$I$ admits a closed form.
To give a partial answer, we recall in Section~\ref{se:d-fin} the definition of the class
of $D$-finite functions.
The main result of this note (Theorem~\ref{thm:main} in Section~\ref{se:main}
below) is that~$I$ is \emph{not} $D$-finite.

We briefly mention some other results of the kind ``a certain
function does not belong to a certain class''. There is a wealth of theorems
about transcendental functions, i.e., non-algebraic ones; for examples, see, e.g.,
Schmidt~\cite{Sc90} and Bell et al.~\cite{BeBrCo12} and the references given there.
Rubel~\cite{Ru89} discusses ``transcendentally transcendental'' functions, which means
that they do not satisfy an algebraic differential equation. The class of
$D$-finite functions, which we are interested in, lies in between the algebraic
and differentially algebraic classes. Finally, Bronstein et al.~\cite{BrCoDaJe08}
have shown that the Lambert $W$ function is not Liouvillian, which roughly
means that it cannot be expressed by iterated integration and exponentiation
starting with an algebraic function.

\section{$D$-finite functions}\label{se:d-fin}

Suppose that a $C^\infty$-smooth function~$f$ is defined
on an open set $\mathcal{D}_f \subseteq \mathbb{R}^n$. It is called
$D$-finite if it satisfies PDEs
\begin{align}
  P_{1,d_1}(\mathbf{x}) \frac{\partial^{d_1}\! f(\mathbf{x})}{\partial x_1^{d_1}}
    + P_{1,d_1-1}(\mathbf{x}) \frac{\partial^{d_1-1}\! f(\mathbf{x})}{\partial x_1^{d_1-1}}
    + \dots +
    P_{1,1}(\mathbf{x}) \frac{\partial f(\mathbf{x})}{\partial x_1}
    + P_{1,0}(\mathbf{x}) f(\mathbf{x})  &= 0, \notag \\
  &\vdots \label{eq:pde} \\
  P_{n,d_n} (\mathbf{x}) \frac{\partial^{d_n}\! f(\mathbf{x})}{\partial x_n^{d_n}}
    + P_{n,d_n-1}(\mathbf{x}) \frac{\partial^{d_n-1}\! f(\mathbf{x})}{\partial x_n^{d_n-1}}
    + \dots +
    P_{n,1}(\mathbf{x}) \frac{\partial f(\mathbf{x})}{\partial x_n}
    + P_{n,0}(\mathbf{x}) f(\mathbf{x})  &= 0,  \notag
\end{align}
valid for $\mathbf{x}=(x_1,\dots,x_n) \in \mathcal{D}_f$, where $d_i\geq 1$
for $i=1,\dots,n$, and the~$P_{ij}$ are polynomials such that
$P_{i,d_i}$ is not identically zero on~$\mathcal{D}_f$ 
for $i=1,\dots,n$. If~$f$ is real analytic,
let us fix an arbitrary point $\mathbf{x}_0\in\mathcal{D}_f$ and consider the Taylor
expansion of~$f$ at~$\mathbf{x}_0$. When viewed as a formal power series, 
the PDEs~\eqref{eq:pde} show that its partial
derivatives generate a finite dimensional vector
space over the field of rational functions. This algebraic definition
of $D$-finiteness
is the usual one in the literature~\cite{Li89,St80}. It cannot be used directly
for \emph{functions}, though: Unlike the field of formal Laurent series, $C^\infty(\mathcal{D}_f)$ is not
a vector space over the field of rational functions.
Recall that in one dimension, the field of formal Laurent series
$R((X))$ in the indeterminate~$X$ consists of the formal series
\[
  \Big\{ \sum_{n\geq n_0} a_n X^n : n_0 \in \mathbb{Z},\
    a_n \in\mathbb{R}, n\geq n_0 \Big\}.
\]
The multivariate case is more involved; see Aparicio Monforte and
Kauers~\cite{ApKa12}.

Rational, and, more generally, algebraic functions are $D$-finite,
as are the elementary functions $\exp$, $\log$, $\sin$, and $\cos$
(but not $\tan$ and $\cot$).
About 60\% of the special functions covered in the classical handbook
by Abramowitz and Stegun~\cite{AbSt64} are $D$-finite.
Examples include Bessel functions, Airy functions, exponential integral, dilogarithm,
and hypergeometric series.
Furthermore, the class of $D$-finite functions is closed under
addition, multiplication, (in-)definite integration, and Laplace
transform.
Division, on the other hand,
does \emph{not} preserve $D$-finiteness in general, nor does composition,
unless the inner function is algebraic. Thus, $x\to e^{\sqrt{x}}$ is
$D$-finite, whereas $x\to e^{e^x}$ and $x\to 1/\sin x$ are not.
The closure under algebraic substitution will be applied
in the proof of our Theorem~\ref{thm:main} below.

$D$-finite functions resp.\ power series have been studied extensively
in the discrete mathematics and symbolic computation literature.
Indeed, provided the coefficients
of the polynomials in~\eqref{eq:pde} and the initial conditions
can be represented with a finite amount of information, one
has a finite data structure that can represent reasonably general
functions (or, in symbolic computation, rather formal power series).
The closure properties mentioned above are effective, in the sense
that there is an algorithm that computes the $D$-finite specification
for a sum of two given $D$-finite functions, etc.
There are also algorithms for proving identities involving
$D$-finite functions, which are about to render tables of such
formulas (certain definite integral evaluations, e.g.)
obsolete to some extent. For further
information, see, e.g., Chyzak and Salvy~\cite{ChSa98} or Koutschan~\cite{Ko09}.

Summarizing, the class of $D$-finite functions is extensive,
but still there are plenty of ``explicit'' functions that are
not $D$-finite. The following theorem thus gives a partial
answer to the question asked in the title of the present note.

\section{Main result}\label{se:main}

\begin{theorem}\label{thm:main}
  The function $I:\mathcal{D}_I \subset \mathbb{R}^3 \to (0,\infty)$
  defined by~\eqref{eq:def I} is not $D$-finite.
\end{theorem}
\begin{proof}
  It is often hard to show directly
  that a multivariate function is not $D$-finite, whereas there are several
  methods available that deal with univariate functions~\cite{BeGeKlLu08,FlGeSa:05,Ge04}.
  Thus, our strategy is to find a useful specialization of~$I$ to a univariate function.
  We let the underlying's initial price be proportional to the strike: $S=eK$.
  For $0<K<1/e$, the inequality
  \[
    (e-1)K < \hat{c}(K) := (e-1)K + e K^2 < eK
  \]
  holds, and so the function
  \[
    f(K) := I(eK,K,\hat{c}(K))
  \]
  is well-defined for $0<K<1/e$. We assume that~$I$ is $D$-finite, and want
  to show that then~$f$ would be $D$-finite, too, by
  studying its Taylor series at $K=1/(2e)$.
  First, we expand the analytic function~$I$
  in a neighborhood of the point
  $(\tfrac12,\tfrac{1}{2e},\tfrac12-\tfrac{1}{4e})\in\mathcal{D}_I$:
  \begin{equation}\label{eq:I ser}
    I(S,K,c)=\sum_{i,j,k\geq0}\gamma_{ijk}(S-\tfrac12)^i(K-\tfrac{1}{2e})^j
      (c-(\tfrac12-\tfrac{1}{4e}))^k.
  \end{equation}
  Since~$I$ is $D$-finite, so is the formal power series
  \[
    \sum_{i,j,k\geq0} \gamma_{ijk} X^i Y^j Z^k
  \]
  in the indeterminates $X,Y,Z$.
  (Recall that this means that its partial derivatives
  generate a finite dimensional subspace of the space of formal Laurent
  series, which is a vector space over the field of rational functions
  $\mathbb{R}(X,Y,Z)$.)
  The algebraic substitution $X\to eX$, $Y\to X$,
  $Z\to(e-1)X+eX^2$ preserves $D$-finiteness~\cite{Ko09,Li89,St80}, and so the univariate
  formal power series
  \begin{equation}\label{eq:f fps}
    \sum_{i,j,k\geq0} (eX)^i X^k (eX+eX^2)^k
  \end{equation}
  is $D$-finite. This series represents the analytic function
  $x\mapsto f(x + \tfrac{1}{2e})$ in a neighborhood of zero. Indeed, by~\eqref{eq:I ser},
  \[
    f(K) = \sum_{i,j,k\geq0}\gamma_{ijk}(eK-\tfrac12)^i(K-\tfrac{1}{2e})^j
      (\hat{c}(K)-(\tfrac12-\tfrac{1}{4e}))^k,
  \]
  and this is~\eqref{eq:f fps} with~$X$ replaced by $K-1/(2e)$. 
  This shows that $D$-finiteness of~$I$
  implies $D$-finiteness of~$f$. 
  
  We will employ the following variant of the Black-Scholes formula
  (see Roper and Rutkowski~\cite{RoRu09}; note that the maturity
  $T>0$ is fixed throughout):
  \begin{equation}\label{eq:roper}
    C_{\mathrm{BS}}(S,K,\sigma) = (S-K)^+ + S \int_0^{\sigma \sqrt{T}}
      N'\Big(\frac{\log(S/K)}{v} + \frac{v}{2}\Big) dv.
  \end{equation}
  Define the function~$F:(0,\infty) \to (0,1/e)$ by
  \begin{equation}\label{eq:F def}
    F(x) := \int_0^x N'\Big(\frac1v+\frac{v}{2}\Big)dv.
  \end{equation}
  Then, for $S=eK$, the log-moneyness is $\log(S/K) = 1$, and
  equation~\eqref{eq:roper} becomes
  \[
    C_{\mathrm{BS}}(eK,K,\sigma) = (e-1)K + eK F(\sigma \sqrt{T}).
  \]
  By the definition of~$f$, we infer that the equation
  \[
    (e-1)K + eK F(\sqrt{T}f(K)) = \hat{c}(K), \quad 0<K<1/e,
  \]
  must hold, and hence
  \[
    F(\sqrt{T}f(K)) = \frac{\hat{c}(K)-(e-1)K}{eK} = K,
  \]
  or
  \[
    f(K) = F^{-1}(K)/\sqrt{T}, \quad 0<K<1/e.
  \]
  The following lemma shows that $F^{-1}$ is not $D$-finite,
  hence~$f$ and~$I$ are not $D$-finite, either.
\end{proof}

\begin{lemma}
  Let the function~$F:(0,\infty) \to (0,1/e)$ be defined by~\eqref{eq:F def}.
  Then the inverse function $F^{-1} : (0,1/e) \to (0,\infty)$ is not $D$-finite.
\end{lemma}
\begin{proof}
  We will show that $F^{-1}$ has a singularity at zero whose type is incompatible
  with $D$-finiteness.
  First, we determine the asymptotics of~$F$ itself as $x \to 0^+$.
  For $0\leq v\leq x$, we have
  \[
    1 \geq e^{-v^2/8} \geq e^{-x^2/8} = 1 + \Oh(x^2),
  \]
  and so
  \begin{equation}\label{eq:n}
    N'\Big(\frac1v+\frac{v}{2}\Big) = \frac{1}{\sqrt{2e\pi}} e^{-1/(2v^2)}
      (1+\Oh(x^2)), \quad x\to 0^+.
  \end{equation}
  Substitution, integration by parts, and the expansion of~$N$ at infinity yield
  \begin{align}
    \frac{1}{\sqrt{2\pi}}\int_0^x e^{-1/(2v^2)} dv
    &= \frac{1}{\sqrt{2\pi}} \int_{1/x}^\infty \frac{e^{-y^2/2}}{y^2}dy \notag \\
    &= \frac{x}{\sqrt{2\pi}} e^{-1/(2x^2)} - \frac{1}{\sqrt{2\pi}}
      \int_{1/x}^\infty e^{-y^2/2} dy \notag \\
    &= \frac{x}{\sqrt{2\pi}} e^{-1/(2x^2)} -1 + N(1/x) \notag \\
    &= \frac{x^3}{\sqrt{2\pi}}e^{-1/(2x^2)}(1+\Oh(x^2)). \label{eq:int}
  \end{align}
  From~\eqref{eq:n} and~\eqref{eq:int} we conclude
  \begin{equation}\label{eq:F asympt}
    F(x) = \frac{x^3}{\sqrt{2e\pi}}e^{-1/(2x^2)}
      (1+\Oh(x^2)), \quad x\to 0^+,
  \end{equation}
  and thus, taking logarithms,
  \begin{equation}\label{eq:F asympt log}
    -\log F(x) = \frac{1}{2x^2} + \Oh\Big(\log \frac1x\Big), \quad x\to 0^+.
  \end{equation}
  Now we study the asymptotics of~$F^{-1}$. For small~$x>0$,
  formula~\eqref{eq:F asympt} yields $F(x) \leq e^{-1/(2x^2)}$.
  Replacing~$x$ by $F^{-1}(x)$ in this inequality and rearranging, we obtain
  \begin{equation}\label{eq:F inv est}
    F^{-1}(x) \leq \frac{1}{\sqrt{2 \log(1/x)}}.
  \end{equation}
  Now we replace~$x$ by $F^{-1}(x)$ in~\eqref{eq:F asympt log} and
  use~\eqref{eq:F inv est} in the $\Oh$-estimate:
  \begin{equation}\label{eq:log F}
    \log \frac1x = \frac{1}{2 F^{-1}(x)^2} + \Oh\Big(\log\log \frac1x\Big).
  \end{equation}
  This implies that~\eqref{eq:F inv est} is sharp:
  \begin{equation}\label{eq:F inv asympt}
    F^{-1}(x) \sim \frac{1}{\sqrt{2 \log(1/x)}}, \quad x\to 0^+.
  \end{equation}
  Observe that a univariate $D$-finite function~$g$ solves an ODE
  \[
    P_d(x) g^{(d)}(x) + \dots + P_1(x) g'(x) + P_0(x) g(x) = 0
  \]
  with polynomial coefficients, where $P_d$ is not identically zero. Such a function
  has an analytic continuation throughout the complex plane, up to a finite number
  of poles and branch points. At each such singularity~$x_0$, an asymptotic expansion
  holds, which may feature powers of~$x-x_0$, exponential terms, and
  \emph{integral} powers of $\log (x-x_0)$.
  This 19th century result~\cite{In26} is a powerful method for proving
  non-$D$-finiteness; see Flajolet et al.~\cite{FlGeSa:05,FlGeSa10} for details.
  In particular, a \emph{fractional} power of a logarithm, as in~\eqref{eq:F inv asympt},
  is a ``forbidden'' asymptotic element, and shows that~$F^{-1}$ is
  not $D$-finite.
\end{proof}

\bibliographystyle{siam}
\bibliography{../gerhold}

\end{document}